\documentclass[conference]{IEEEtran}

\usepackage{algorithm}
\usepackage{algorithmic}
\usepackage{amsmath,amssymb,amsthm,bm}
\usepackage{mathtools}
\usepackage{graphicx}
\usepackage{booktabs}
\usepackage{multirow}
\usepackage{hyperref}
\usepackage{enumitem}
\usepackage{cite}
\usepackage{float}
\usepackage{braket}
\usepackage{caption}
\usepackage[framemethod=TikZ]{mdframed}
\usepackage{xcolor}
\usepackage{orcidlink}

\hypersetup{
  breaklinks=true,
  colorlinks=true,
  linkcolor=blue,
  urlcolor=blue,
  citecolor=blue
}

\graphicspath{{./}}

\newcommand{\one}{\mathbf{1}}
\newcommand{\zall}{\vec{z}^{(\one)}}
\newcommand{\zstar}{\vec{z}^{\star}}
\newcommand{\ztrunc}{\vec{z}^{\star\star}}
\newcommand{\bitflip}[1]{\vec{z}^{\,\oplus #1}}
\newcommand{\kmax}{k_{\max}}
\newcommand{\kfull}{k_{\mathrm{full}}}
\newcommand{\Nomit}{N_{>}}

\theoremstyle{plain}
\newtheorem{theorem}{Theorem}[section]
\newtheorem{corollary}{Corollary}[theorem]
\newtheorem{proposition}{Proposition}
\newtheorem{lemma}[proposition]{Lemma}

\theoremstyle{remark}

\theoremstyle{definition}

\newmdtheoremenv[
    linewidth=0.5pt,
    linecolor=black,
    backgroundcolor=black!3,
    innertopmargin=6pt,
    innerbottommargin=6pt,
    innerleftmargin=8pt,
    innerrightmargin=8pt,
    roundcorner=2pt,
    skipabove=6pt,
    skipbelow=6pt
]{definition}[proposition]{Definition}

\DeclareMathOperator{\E}{\mathbb{E}}
\DeclareMathOperator{\Var}{Var}

\def\BibTeX{{\rm B\kern-.05em{\sc i\kern-.025em b}\kern-.08em
    T\kern-.1667em\lower.7ex\hbox{E}\kern-.125emX}}

\title{Truncated-Binary Encoding: \\ Spectral Degree Reduction of Combinatorial Optimization Problems for Quantum Hardware}

\author{
\IEEEauthorblockN{Tristan Zaborniak\textsuperscript{1,2} \orcidlink{0000-0002-4301-0861}}
\IEEEauthorblockA{\textsuperscript{1}\textit{Department of Computer Science, University of Victoria}, 3800 Finnerty Rd., Victoria, BC V8P 5C2, Canada\\
\textsuperscript{2}\textit{Center for Computational Biology, Flatiron Institute}, 162 Fifth Avenue, New York, NY 10010, USA}
}

\begin{document}

\maketitle
\thispagestyle{plain}
\pagestyle{plain}

\begin{abstract}

    Exact-binary encoding compiles a discrete cost function network (CFN) into a higher-order unconstrained binary optimization (HUBO) problem whose maximum monomial degree grows with the cardinalities of the underlying CFN variables. Given that quantum optimization hardware generally favours quadratic unconstrained binary optimization or low-degree HUBO Hamiltonians, high-cardinality CFNs therefore incur substantial overhead in the form of circuit depth, or ancilla qubits when degree-reduction techniques are employed. To ameliorate these issues, we propose \textit{truncated-binary encoding} (TBE): a modification of exact-binary encoding in which Ising-basis monomials exceeding a chosen cutoff $\kmax$ are dropped from the encoded cost. We establish a tight $L^\infty$ bound on the truncation error in terms of the omitted couplings, derive sufficient conditions on the energy gap and on the single bit-flip basin barrier under which TBE preserves the global minimum and its local-minimum structure, and characterize a noise floor condition on the spectral profile under which the truncation residual acts as a perturbative correction to the underlying landscape. We then express the encoded coefficients directly as Walsh transforms of the underlying CFN cost tables, and prove a bound under which smoothness of each cost table implies rapid decay of its high-degree Walsh mass. Together these results yield a principled \textit{a priori} criterion for selecting $\kmax$ and for judging whether a given CFN admits a small-$\kmax$ TBE.

\end{abstract}

\smallskip


\section{Introduction}
\label{sec:intro}

Quantum optimization hardware is most naturally suited to quadratic or low-degree polynomial Hamiltonians. In annealing hardware, native interactions are essentially $2$-local \cite{Johnson2011, Albash2018}, so any interaction of higher order must be reduced to pairwise form. In gate-based optimization---such as the quantum approximate optimization algorithm (QAOA) and its alternating-operator generalizations \cite{Farhi2014, Hadfield2019}---native gates are likewise low-degree, and $k$-body Pauli interactions translate into per-layer circuit depths that scale linearly in $k$. Combinatorial optimization problems over discrete variables must therefore be encoded and reduced to quantum hardware-compatible forms before they can be solved. This is typically accomplished by introducing ancilla qubits and quadratization gadgets \cite{Boros2002, Dattani2019, Anthony2017}. Either route incurs nontrivial overhead, often dominating the resource requirements of the original problem.

Many practical combinatorial optimization problems are most cleanly expressed as cost function networks (CFNs), in which $N$ discrete variables $\vec{d} = (d_1,\ldots,d_N)$ with cardinalities $|d_i|$ interact through cost tables \cite{Schiex1995, Cooper2010}. Encoding such problems for quantum hardware requires a choice of binary representation \cite{Lucas2014}. The most popular options are one-hot encoding (which uses $\sum_i |d_i|$ binary variables and produces a QUBO) \cite{Lucas2014}, domain-wall encoding (which uses $\sum_i (|d_i|-1)$ variables and likewise produces a QUBO) \cite{Chancellor2019}, and exact binary encoding (which uses only $\sum_i \lceil \log_2 |d_i| \rceil$ variables but produces a HUBO of degree up to $\max_{(i,j)}\bigl(\lceil \log_2 |d_i|\rceil + \lceil \log_2 |d_j|\rceil\bigr)$) \cite{Berwald2024}. Exact binary encoding is therefore the most qubit-efficient option prior to quadratization or degree-reduction, but generates the most hardware-unfriendly Hamiltonian.

A complementary compilation axis---approximation quality---has been opened more recently by approximate-binary encoding \cite{Zaborniak2025}, which retains the logarithmic qubit count of exact binary encoding, but fits the HUBO couplings at a chosen target degree by solving an overdetermined least-squares system, so that the encoded cost approximates the underlying CFN rather than reproducing it exactly. We propose a different route along the same axis. The \textit{truncated-binary encoding} (TBE) at degree $\kmax$ is the polynomial obtained from exact binary encoding by discarding all monomials of degree exceeding $\kmax$  in the Ising basis. The resulting HUBO is hardware-targetable: choosing $\kmax = 2$ recovers a QUBO directly without ancillas, and choosing $\kmax = 3$ or $4$ matches the native locality of certain QAOA implementations or LHZ-style architectures \cite{Lechner2015}. The trade-off is approximation error, and the central question of this paper is when this error is small enough to preserve the global optimum and its basin.

We adopt the Ising basis $\{-1,+1\}^n$ throughout. Although exact binary encoding is most often presented in the $\{0,1\}^n$ basis \cite{Berwald2024}, in the Ising basis the monomial $\prod_{i\in S} z_i$ is identically the Walsh--Hadamard function $\chi_S$ \cite{ODonnell2014, deWolf2008}. As a result, the polynomial-degree expansion of an Ising HUBO is precisely its Walsh expansion, and the HUBO degree of every term coincides with its Walsh spectral degree: HUBO-degree truncation is then the orthogonal projection onto the Walsh subspace of degree at most $\kmax$, removing precisely the short-correlation-length, high-influence modes responsible for landscape ruggedness while preserving long-correlation-length basin structure. The same identification fails in the $\{0,1\}^n$ basis, where every $\{0,1\}$-form monomial of degree $k$ contributes Walsh content at all degrees $0, 1, \ldots, k$.

We establish four results. First, the truncation error in the $L^\infty$ sense is tightly bounded by the sum of the magnitudes of the omitted couplings (Theorem \ref{thm:linfty}); this bound is computable in time linear in the number of HUBO monomials and provides a direct, certificate-style guarantee. Second, when the energy gap of the original CFN exceeds twice this bound, the global minimum of the truncated encoding coincides with that of the full encoding (Theorem \ref{thm:optimum_preservation}), and an analogous statement holds for the local-minimum status of $\zstar$ under single bit-flip moves (Theorem \ref{thm:basin_preservation}). Third, the spectral profile $\{P_k\}$ of the Ising couplings provides a quantitative diagnostic for the appropriateness of a given cutoff (Section \ref{subsec:gaussian}). Fourth, this aggregate spectral profile decomposes additively over the cost tables of the underlying CFN, and CFN-level structural features allow \textit{a priori} provability of front-loaded spectra under appropriate bitstring assignments.

The remainder of this paper is organized as follows. Section \ref{sec:preliminaries} introduces exact binary encoding directly in the Ising basis, reviews the Walsh--Hadamard analysis required for the spectral perspective, and establishes a spectral leakage formula that justifies the Ising convention. Section \ref{sec:tbe} formally defines TBE and establishes the $L^\infty$ error bound. Section \ref{sec:preservation} derives sufficient conditions on the energy gap for global-minimum and basin preservation. Section \ref{sec:spectral} further develops the spectral perspective, deriving a noise floor condition. Section \ref{sec:cfn-structure} traces the spectrum back to the underlying cost tables, and develops the discrete-calculus machinery needed for smoothness criteria favourable to front-loaded spectra and TBE. Finally, Section \ref{sec:algorithm} presents the practical compilation algorithm, and we close with a discussion and conclusion in Section \ref{sec:discussion}.

\section{Preliminaries: Exact-Binary Ising Encoding and Walsh--Hadamard Spectral Decomposition}
\label{sec:preliminaries}

\subsection{Cost Function Networks}
\label{subsec:cfn}

A pairwise-decomposable CFN over $N$ discrete variables with cardinalities $|d_i|$ specifies a cost:

\begin{equation}\label{eq:cfn}
    f_{\mathrm{CFN}}(\vec{d}) = \sum_{i=1}^{N} \alpha_i(d_i) + \sum_{1\leq i < j \leq N} \beta_{i,j}(d_i, d_j)
\end{equation}

\noindent where $\alpha_i\colon \{1,\ldots,|d_i|\} \to \mathbb{R}$ and $\beta_{i,j}\colon \{1,\ldots,|d_i|\} \times \{1,\ldots,|d_j|\} \to \mathbb{R}$ are tabulated. Throughout, we refer to $\alpha_i$ as the \textit{unary cost table} for variable $i$ and to $\beta_{i,j}$ as the \textit{pairwise cost table} for the unordered variable pair $\{i,j\}$. The CFN induces an interaction graph having the vertex set $\{1,\ldots,N\}$ and edges $\{i,j\}$ whenever $\beta_{i,j} \not\equiv 0$.

The combinatorial optimization problem associated with a CFN is to find $\vec{d}^\star \in \arg\min_{\vec{d}} f_{\mathrm{CFN}}(\vec{d})$. Many problems of practical interest fit this template: weighted constraint satisfaction \cite{Schiex1995}, weighted Max-$k$-SAT, graph coloring \cite{Lucas2014}, and rotamer-based protein design \cite{Allouche2014, Babbush2014} are examples. While we work with pairwise-decomposable CFNs throughout, our method generalizes to CFNs of higher order.

\subsection{Exact Binary Encoding in the Ising Basis}
\label{subsec:ising-encoding}

Following Ref. \cite{Berwald2024}, exact binary encoding assigns to each index-$i$ variable choice $c_i \in \{1,\ldots,|d_i|\}$ a distinct bitstring $\vec{r}^{\,(c_i)} \in \{0,1\}^{D_i}$ of length:

\begin{equation}\label{eq:Di}
    D_i = \lceil \log_2 |d_i| \rceil
\end{equation}

\noindent When $|d_i|$ is not a power of two, exactly $2^{D_i} - |d_i|$ of the available bitstrings are unused; we treat their cost contributions as either undefined (to be handled by penalty) or set to a fallback choice value, deferring discussion of this point to Section \ref{sec:algorithm}. The bitstring assignment $c_i \mapsto \vec{r}^{\,(c_i)}$ is a free parameter of the encoding; we will see in Section \ref{sec:cfn-structure} that its choice substantially affects the spectral profile of the encoded HUBO and that it must be chosen jointly with $\kmax$.

We work in Ising variables throughout, with one Ising spin per bit. For register $i$, we write $\vec{z}_i = (z_{i,0},\ldots,z_{i,D_i-1}) \in \{-1,+1\}^{D_i}$ and adopt the bit-to-spin convention:

\begin{equation}\label{eq:bit-to-spin}
    z_{i,q} = 1 - 2\,b_{i,q} \quad \text{for } q = 0, \ldots, D_i - 1
\end{equation}

\noindent where $b_{i,q}\in\{0,1\}$, so that bit $0$ corresponds to spin $+1$ and bit $1$ to spin $-1$. The complete spin vector is $\vec{z} = (\vec{z}_1, \ldots, \vec{z}_N) \in \{-1,+1\}^n$, and the total qubit count is:

\begin{equation}\label{eq:n-total}
    n = \sum_{i=1}^{N} D_i
\end{equation}

\noindent Further, to each choice $c_i$ corresponds a sign vector:

\begin{equation}\label{eq:sign-vector}
    \vec{s}^{\,(c_i)} = \bigl(1 - 2 r_0^{(c_i)},\, \ldots,\, 1 - 2 r_{D_i-1}^{(c_i)}\bigr) \in \{-1,+1\}^{D_i}
\end{equation}

\noindent which records the spin pattern that encodes choice $c_i$.

The indicator polynomial that selects choice $c_i$ from register $i$ is the polynomial in $\vec{z}_i$ that evaluates to $1$ when $\vec{z}_i = \vec{s}^{\,(c_i)}$ and to $0$ on every other Ising configuration. Constructing this product literal-by-literal:

\begin{equation}\label{eq:indicator}
    x_{i,c_i}(\vec{z}_i) = \frac{1}{2^{D_i}} \prod_{q=0}^{D_i - 1} \bigl(1 + s_q^{(c_i)}\,z_{i,q}\bigr)
\end{equation}

\noindent each factor evaluates to $2$ when $z_{i,q} = s_q^{(c_i)}$ (the spin matches) and to $0$ otherwise; the product therefore equals $2^{D_i}$ when all bits match and $0$ when any single bit disagrees, giving the indicator after the prefactor $2^{-D_i}$. Expanding yields:

\begin{equation}\label{eq:indicator_walsh}
    x_{i,c_i}(\vec{z}_i) = \frac{1}{2^{D_i}} \sum_{T \subseteq [D_i]} \sigma_T^{(c_i)}\,\prod_{q\in T} z_{i,q}
\end{equation}

\noindent where:

\begin{equation}\label{eq:sigma-def}
    \sigma_T^{(c_i)} =\prod_{q \in T} s_q^{(c_i)} \in \{-1,+1\}
\end{equation}

\noindent is the product of the spin-determined signs over the index set $T$, and $[D_i] = \{0, 1, \ldots, D_i - 1\}$. Equation \eqref{eq:indicator_walsh} is already the Walsh expansion of $x_{i,c_i}$ on the register-$i$ hypercube; this is the key consequence of working in the Ising basis from the start, and we exploit it throughout.

Because exactly one $x_{i,c_i}(\vec{z}_i)$ takes the value $1$ for any $\vec{z}_i$ that encodes a valid choice (the others being $0$), the unary cost contribution at register $i$ can be written as:

\begin{equation}\label{eq:unary-encoded}
    \alpha_i\bigl(d_i(\vec{z}_i)\bigr) = \sum_{c_i = 1}^{|d_i|} \alpha_i(c_i)\,x_{i,c_i}(\vec{z}_i)
\end{equation}

\noindent and the pairwise contribution analogously as:

\begin{equation}\label{eq:binary-encoded}
    \beta_{i,j}(d_i, d_j) = \sum_{c_i, c_j} \beta_{i,j}(c_i, c_j)\,x_{i,c_i}(\vec{z}_i)\,x_{j,c_j}(\vec{z}_j)
\end{equation}

\noindent Substituting Equations~\eqref{eq:unary-encoded}--\eqref{eq:binary-encoded} into Equation \eqref{eq:cfn} and collecting like monomials yields the Ising HUBO form:

\begin{equation}\label{eq:ising_hubo}
    f(\vec{z}) = \sum_{S \subseteq [n]} c_S \prod_{i \in S} z_i
\end{equation}

\noindent of maximum degree:

\begin{equation}\label{eq:max_degree}
    \kfull = \max_{(i,j)} (D_i + D_j)
\end{equation}

\noindent attained by products of two indicator polynomials, each contributing $D_i$ literal variables. The couplings $\{c_S\}$ are explicit linear combinations of the CFN cost-table entries $\{\alpha_i(c_i),\,\beta_{i,j}(c_i, c_j)\}$ weighted by the signs $\sigma_T^{(c_i)}$ of Equation \eqref{eq:sigma-def} and powers of $1/2$ from Equation \eqref{eq:indicator_walsh}; their precise form is derived in Section \ref{subsec:cfn-walsh}.

Equation \eqref{eq:max_degree} is the obstruction we ultimately wish to address. For a CFN with two interacting variables of cardinality $32$ (so $D_i = D_j = 5$), $\kfull = 10$; for cardinality $128$, $\kfull = 14$. Both are well outside the native locality of contemporary quantum optimization hardware.

\subsection{Walsh--Hadamard Analysis on the Ising Hypercube}
\label{subsec:walsh}

The Walsh--Hadamard basis provides a discrete analogue of Fourier analysis on the hypercube and is the natural setting for spectral decomposition arguments in this setting \cite{Hammer1968, ODonnell2014, deWolf2008, Weinberger1990}. For each $S \subseteq [n]$, the Walsh function $\chi_S\colon \{-1,+1\}^n \to \{-1,+1\}$ is defined by:

\begin{equation}\label{eq:walsh-def}
    \chi_S(\vec{z}) = \prod_{i \in S} z_i
\end{equation}

\noindent and has \textit{degree} $|S|$. The family $\{\chi_S\}_{S \subseteq [n]}$ is a complete orthonormal basis under the uniform measure on $\{-1,+1\}^n$:

\begin{equation}\label{eq:walsh-ortho}
    \E_{\vec{z}}\bigl[\chi_S(\vec{z})\,\chi_T(\vec{z})\bigr] = \delta_{S,T}
\end{equation}

\noindent where the expectation is taken with respect to the uniform distribution. Any $f\colon \{-1,+1\}^n \to \mathbb{R}$ admits a unique expansion:

\begin{equation}\label{eq:walsh-exp}
    f(\vec{z}) = \sum_{S \subseteq [n]} \hat{f}(S)\,\chi_S(\vec{z})
\end{equation}

\noindent with Walsh coefficients $\hat{f}(S) = \E_{\vec{z}}[f(\vec{z})\,\chi_S(\vec{z})]$. The variance of $f$ decomposes by degree:

\begin{equation}\label{eq:variance-decomp}
    \sigma_f^2 = \sum_{k \geq 1} P_k \quad \text{with} \quad P_k = \sum_{|S| = k} \hat{f}(S)^2
\end{equation}

\noindent the \textit{degree-$k$ spectral power}.

Note that the Ising HUBO of Equation \eqref{eq:ising_hubo} is identically its Walsh expansion. Comparing Equations~\eqref{eq:ising_hubo} and~\eqref{eq:walsh-exp}:

\begin{equation}\label{eq:walsh_identification}
    f(\vec{z}) = \sum_{S\subseteq[n]} c_S \prod_{i\in S} z_i = \sum_{S\subseteq[n]} \hat{f}(S)\,\chi_S(\vec{z})
\end{equation}

\noindent so $\hat{f}(S) = c_S$ exactly, with no further calculation. The polynomial degree of any monomial coincides with its Walsh-spectral degree, and the spectral profile $\{P_k\}$ is computable in $O(m)$ time from the $m$ nonzero couplings:

\begin{equation}\label{eq:spectral_profile}
    P_k = \sum_{|S|=k} c_S^2
\end{equation}

\noindent without further basis transformation. Every operation we perform in the polynomial basis has an immediate spectral interpretation, and, conversely, every spectral construction has an immediate polynomial-basis implementation. We exploit both directions in what follows.

\subsection{Spectral Leakage and the Inequivalence of the $\{0,1\}$ Basis}
\label{subsec:leakage}

The substitution $b_i = (1 - z_i)/2$ converts an Ising HUBO into an equivalent $\{0,1\}$-form HUBO and vice versa: the two encodings represent the same cost function on the same configuration space and yield the same global optimum. They are \textit{not} equivalent for purposes of truncation, however, and the inequivalence is enough to warrant a short proof.

\begin{theorem}[Spectral leakage in the $\{0,1\}$ basis]\label{prop:leakage}

    Let $f\colon \{-1,+1\}^n \to \mathbb{R}$ admit a $\{0,1\}$-form polynomial HUBO expansion:

    \begin{equation}\label{eq:01-expansion}
        f(\vec{b}) = \sum_{S \subseteq [n]} c^\prime_S \prod_{i \in S} b_i \quad \text{with} \quad b_i = \tfrac{1}{2}(1 - z_i)
    \end{equation}

    \noindent Then the Walsh coefficient of $f$ at $T$ is:

    \begin{equation}\label{eq:leakage-formula}
        \hat{f}(T) = (-1)^{|T|} \sum_{S \supseteq T} \frac{c^\prime_S}{2^{|S|}}
    \end{equation}

    \noindent In particular, every $\{0,1\}$-monomial of degree $k$ contributes Walsh mass at every degree $0, 1, \ldots, k$.

\end{theorem}

\begin{proof}

    For any $S \subseteq [n]$:

    \begin{equation}\label{eq:b-product}
        \prod_{i \in S} b_i = \prod_{i \in S} \frac{1 - z_i}{2} = \frac{1}{2^{|S|}} \sum_{T \subseteq S} (-1)^{|T|}\,\chi_T(\vec{z})
    \end{equation}

    \noindent obtained by expanding the product $\prod_{i \in S}(1 - z_i)$. Substituting Equation \eqref{eq:b-product} into Equation \eqref{eq:01-expansion} and exchanging the order of summation, we get:

    \begin{equation}\label{eq:double-sum}
        f(\vec{z}) = \sum_{T \subseteq [n]} (-1)^{|T|}\,\chi_T(\vec{z}) \sum_{S \supseteq T} \frac{c^\prime_S}{2^{|S|}}
    \end{equation}

    \noindent Reading off the coefficient of $\chi_T$ gives Equation \eqref{eq:leakage-formula}.

\end{proof}

The implication for truncation is immediate. Suppose we truncate $f$ in the $\{0,1\}$ basis by setting $c^\prime_S = 0$ for $|S| > \kmax$. By Equation \eqref{eq:leakage-formula}, the Walsh coefficient $\hat{f}(T)$ for any $T$ with $|T| \leq \kmax$ changes by:

\begin{equation}\label{eq:leakage-error}
    (-1)^{|T|} \sum_{\substack{S \supseteq T \\ |S| > \kmax}} \frac{c^\prime_S}{2^{|S|}}
\end{equation}

\noindent which is generally nonzero. Truncation in the $\{0,1\}$ basis therefore perturbs every Walsh band below the cutoff. By contrast, in the Ising basis $\hat{f}(T) = c_T$ exactly (Equation \eqref{eq:walsh_identification}), so dropping $c_S$ with $|S| > \kmax$ leaves $\hat{f}(T)$ unchanged for every $|T| \leq \kmax$.

\section{Truncated-Binary Encoding}
\label{sec:tbe}

\subsection{Definition of the Truncated Binary Encoding}
\label{subsec:tbe-def}

\begin{definition}[Truncated-binary encoding]\label{def:tbe}

    Let $f$ be the Ising HUBO of Equation \eqref{eq:ising_hubo} obtained from a CFN $f_{\mathrm{CFN}}$ by exact binary encoding. The \textit{truncated-binary encoding} (TBE) of $f_{\mathrm{CFN}}$ at cutoff $\kmax \in \{1,\ldots,\kfull\}$ is:

    \begin{equation}\label{eq:tbe}
        f_{\leq \kmax}(\vec{z}) = \sum_{\substack{S \subseteq [n] \\ |S| \leq \kmax}} c_S \prod_{i \in S} z_i
    \end{equation}

    \noindent obtained from $f$ by deleting all monomials of degree exceeding $\kmax$. Equivalently, Equation \eqref{eq:tbe} may be obtained by orthogonal projection onto the Walsh subspace:

    \begin{equation}\label{eq:Vk}
        V_{\leq \kmax} = \mathrm{span}\bigl\{\chi_S : |S| \leq \kmax\bigr\}
    \end{equation}

\end{definition}

The truncated polynomial $f_{\leq \kmax}$ is, in general, no longer the exact binary encoding of any CFN. The structural identity $\sum_{c_i = 1}^{|d_i|} x_{i,c_i}(\vec{z}_i) \in \{0,1\}$ depends on cancellations between high-degree monomials in the literal-product expansion of Equation \eqref{eq:indicator}; truncation removes those high-degree monomials and breaks the identity. The TBE is therefore a low-pass-filter \textit{approximation} of the Ising cost function $f$, evaluated on the full hypercube $\{-1,+1\}^n$. We now find the $L^\infty$ bound of the resulting approximation error.


\subsection{The $L^\infty$ Truncation Error Bound}
\label{subsec:linfty}

The supremum norm ($L^\infty$-norm) on the hypercube is:

\begin{equation}\label{eq:sup-norm}
    \|f\|_\infty = \max_{\vec{z} \in \{-1,+1\}^n} |f(\vec{z})|
\end{equation}

\noindent and bounds the worst-case pointwise error between two functions on the hypercube. We bound $\|f - f_{\leq \kmax}\|_\infty$ in terms of the omitted couplings.

\begin{theorem}[$L^\infty$ truncation error bound]\label{thm:linfty}

    Let $f$ and $f_{\leq \kmax}$ be as in Definition~\ref{def:tbe}, and define:

    \begin{equation}\label{eq:epsilon_def}
        \varepsilon(\kmax) \coloneqq \sum_{\substack{S \subseteq [n] \\ |S| > \kmax}} |c_S|
    \end{equation}

    \noindent Then:

    \begin{equation}\label{eq:linfty_bound}
        \|f - f_{\leq \kmax}\|_\infty \leq \varepsilon(\kmax)
    \end{equation}

    \noindent and the bound is attained whenever there exists $\vec z \in \{-1,+1\}^n$ with $\chi_S(\vec z) = \operatorname{sign}(c_S)$ for every $S$ with $|S| > \kmax$. In particular, the all-ones configuration $\zall = (+1,\ldots,+1)$ saturates the bound whenever the omitted couplings $\{c_S : |S| > \kmax\}$ all share a common sign,
    since $\chi_S(\zall) = +1$ for every $S$.

\end{theorem}

\begin{proof}

    The pointwise difference is:

    \begin{equation}\label{eq:proof-diff}
        f(\vec{z}) - f_{\leq \kmax}(\vec{z}) = \sum_{|S| > \kmax} c_S\,\chi_S(\vec{z})
    \end{equation}

    \noindent Each Walsh function takes values $\chi_S(\vec{z}) \in \{-1,+1\}$, so by the triangle inequality:

    \begin{equation}\label{eq:proof-tri}
        \Bigl|\sum_{|S| > \kmax} c_S\,\chi_S(\vec{z})\Bigr| \leq \sum_{|S| > \kmax} |c_S|\,|\chi_S(\vec{z})| = \sum_{|S| > \kmax} |c_S|
    \end{equation}

    \noindent which is Equation \eqref{eq:linfty_bound}. At $\zall$, every Walsh function evaluates to $\chi_S(\zall) = \prod_{i \in S}(+1) = 1$, so the difference equals $\sum_{|S|>\kmax} c_S$; when all $c_S$ share a common sign this magnitude is exactly $\varepsilon(\kmax)$.

\end{proof}

The bound of Equation \eqref{eq:linfty_bound} is computable in time linear in the number of HUBO monomials and is therefore available as a certificate at compile time, before any optimization is run. In practice, for randomly distributed coupling signs, the typical pointwise error is much smaller than $\varepsilon(\kmax)$ owing to cancellations among Walsh modes of opposite sign. Section \ref{sec:spectral} replaces the worst-case $L^\infty$ analysis with a typical-case $L^2$ analysis and derives a quantitatively tighter spectral diagnostic, but we first treat in further detail the $L^\infty$ case in Section \ref{sec:preservation}, deriving exact bounds on global and local minima preservation after cost function truncation.

\section{Optimum and Basin Preservation}
\label{sec:preservation}

The $L^\infty$ bound translates directly into preservation guarantees once it is compared against the energy gap of the original cost function. We treat the global-minimum guarantee first, then a strictly weaker but more easily satisfied local-minimum guarantee that controls basin structure under single bit-flips.

\subsection{Global Minimum Preservation}
\label{subsec:global}

\begin{theorem}[Global minimum preservation under truncation]\label{thm:optimum_preservation}

    Let $\zstar \in \arg\min_{\vec{z}} f(\vec{z})$ be a global minimizer of the full Ising HUBO and define the energy gap:

    \begin{equation}\label{eq:gap_def}
        \Delta E^\star = \min_{\vec{z} \notin \arg\min f} f(\vec{z}) - f(\zstar)
    \end{equation}

    \noindent If $\Delta E^\star > 2\varepsilon(\kmax)$, then every minimizer of $f_{\leq \kmax}$ is also a minimizer of $f$.

\end{theorem}

\begin{proof}

    Abbreviate $\varepsilon = \varepsilon(\kmax)$, and let $\ztrunc \in \arg\min f_{\leq \kmax}$ be any minimizer of the truncated encoding. We aim to show that $\ztrunc \in \arg\min f$ as well, so that every minimizer of $f_{\leq \kmax}$ is also a minimizer of the full encoding.

    Theorem \ref{thm:linfty} bounds the supremum norm of the truncation residual by $\varepsilon$, which gives us the following two-sided estimate:

    \begin{equation}\label{eq:two_sided}
        f_{\leq \kmax}(\vec{z}) - \varepsilon \leq f(\vec{z}) \leq f_{\leq \kmax}(\vec{z}) + \varepsilon
    \end{equation}

    \noindent valid for every $\vec{z} \in \{-1,+1\}^n$. Now, apply the upper bound of Equation \eqref{eq:two_sided} at $\vec{z} = \ztrunc$, invoke the optimality of $\ztrunc$ for $f_{\leq \kmax}$, and apply the lower bound of Equation \eqref{eq:two_sided} at $\vec{z} = \zstar$:

    \begin{align}\label{eq:bracket}
        f(\ztrunc)
        &\leq f_{\leq \kmax}(\ztrunc) + \varepsilon \nonumber \\
        &\leq f_{\leq \kmax}(\zstar) + \varepsilon \nonumber \\
        &\leq \bigl(f(\zstar) + \varepsilon\bigr) + \varepsilon \nonumber \\
        &= f(\zstar) + 2\varepsilon
    \end{align}

    \noindent The first inequality is the upper bound of Equation \eqref{eq:two_sided} evaluated at $\ztrunc$. The second follows because $\ztrunc \in \arg\min f_{\leq \kmax}$ implies $f_{\leq \kmax}(\ztrunc) \leq f_{\leq \kmax}(\vec{z})$ for every $\vec{z}$, including $\vec{z} = \zstar$. The third inequality is the lower bound of Equation \eqref{eq:two_sided} evaluated at $\zstar$ and rearranged as $f_{\leq \kmax}(\zstar) \leq f(\zstar) + \varepsilon$.

    Subtracting $f(\zstar)$ from both ends of Equation \eqref{eq:bracket} and applying the hypothesis $\Delta E^\star > 2\varepsilon$:

    \begin{equation}\label{eq:excess_bound}
        f(\ztrunc) - f(\zstar) \leq 2\varepsilon < \Delta E^\star
    \end{equation}

    \noindent It remains to deduce $\ztrunc \in \arg\min f$ from this excess bound. By Equation \eqref{eq:gap_def}, any non-minimizer $\vec{z} \notin \arg\min f$ satisfies $f(\vec{z}) \geq f(\zstar) + \Delta E^\star$; equivalently, by contrapositive, any $\vec{z}$ with $f(\vec{z}) - f(\zstar) < \Delta E^\star$ must lie in $\arg\min f$. Equation \eqref{eq:excess_bound} places $\ztrunc$ in the latter category, so $\ztrunc \in \arg\min f$. Since $\ztrunc$ was an arbitrary minimizer of $f_{\leq \kmax}$, the conclusion holds for every such minimizer.

\end{proof}

The condition $\Delta E^\star > 2\varepsilon(\kmax)$ is among the central design constraints for selecting $\kmax$. However, the energy gap $\Delta E^\star$ is a property of the original CFN and is, in the worst case, NP-hard to compute exactly; in practice it is typically estimated either by problem-specific theoretical bounds or empirically by sampling. Equation \eqref{eq:bracket} additionally yields a useful weaker statement that does not require knowledge of the gap.

\begin{corollary}[Approximate-minimum recovery]\label{cor:approx}

    For any $\kmax$, every minimizer $\ztrunc$ of $f_{\leq \kmax}$ satisfies:

    \begin{equation}\label{eq:approx_bound}
        f(\ztrunc) \leq f(\zstar) + 2\varepsilon(\kmax)
    \end{equation}

    \noindent and so achieves the global cost up to additive error $2\varepsilon(\kmax)$ even when the gap condition of Theorem \ref{thm:optimum_preservation} fails.

\end{corollary}

Equation \eqref{eq:approx_bound} establishes TBE as a $2\varepsilon(\kmax)$-additive approximation scheme, with $\kmax$ tuneable along a smooth quality--cost frontier ranging from $\kmax = \kfull$ to $\kmax = 2$.

\subsection{Local Minimum Preservation}
\label{subsec:basin}

A decidedly weaker but still practically useful guarantee concerns the persistence of $\zstar$ as a local minimum of $f_{\leq \kmax}$ under single bit-flip moves.

\begin{definition}[Single bit-flip local minimum and basin barrier]\label{def:basin}

    A configuration $\vec{z} \in \{-1,+1\}^n$ is a \textit{single bit-flip local minimum} of $f\colon \{-1,+1\}^n \to \mathbb{R}$ if $f(\vec{z}) \leq f(\bitflip{i})$ for every $i \in [n]$, where $\bitflip{i}$ denotes $\vec{z}$ with coordinate $i$ negated. Its \textit{basin barrier} at $\vec{z}$ is:

    \begin{equation}\label{eq:barrier}
        \delta_f(\vec{z}) \coloneqq \min_{i \in [n]} \bigl[f(\bitflip{i}) - f(\vec{z})\bigr]
    \end{equation}

\end{definition}

The basin barrier is the smallest single bit-flip energy increase available from $\vec{z}$. A local minimum of $f$ remains so under $f' = f + \eta$ with $\|\eta\|_\infty < \delta_f(\vec{z})/2$; the perturbation can shift each neighbouring energy by at most this amount.

\begin{theorem}[Local minimum preservation under truncation]\label{thm:basin_preservation}

    Let $\zstar$ be a single bit-flip local minimum of $f$ with basin barrier $\delta_f(\zstar)$. If $\delta_f(\zstar) > 2\varepsilon(\kmax)$, then $\zstar$ is also a single bit-flip local minimum of $f_{\leq \kmax}$, with basin barrier:

    \begin{equation}\label{eq:barrier_loss}
        \delta_{f_{\leq \kmax}}(\zstar) \geq \delta_f(\zstar) - 2\varepsilon(\kmax)
    \end{equation}

\end{theorem}

\begin{proof}

    Abbreviate $\varepsilon = \varepsilon(\kmax)$ and fix $i \in [n]$. We lower-bound the single bit-flip energy difference $f_{\leq \kmax}(\vec{z}^{\star \oplus i}) - f_{\leq \kmax}(\zstar)$ by applying the two halves of Equation \eqref{eq:two_sided} in opposite directions to absorb the worst case in which truncation simultaneously raises the energy of $\zstar$ and lowers that of its neighbour by $\varepsilon$ each. Subtracting:

    \begin{align}\label{eq:basin-proof}
        f_{\leq \kmax}(\vec{z}^{\star \oplus i}) - f_{\leq \kmax}(\zstar)
        &\geq \bigl(f(\vec{z}^{\star \oplus i}) - f(\zstar)\bigr) - 2\varepsilon \nonumber \\
        &\geq \delta_f(\zstar) - 2\varepsilon
    \end{align}

    \noindent where the second inequality is Definition~\ref{def:basin}. By hypothesis this is positive, so $\zstar$ is a single bit-flip local minimum of $f_{\leq \kmax}$. Taking the minimum over $i$ gives Equation \eqref{eq:barrier_loss}.

\end{proof}

Theorem \ref{thm:basin_preservation} asserts that the global optimum, viewed as a fixed point of single bit-flip descent on $f$, remains a fixed point of the same descent on $f_{\leq \kmax}$ provided the basin barrier exceeds $2\varepsilon(\kmax)$. The relationship between $\Delta E^\star$ and $\delta_f(\zstar)$ determines which of Theorems~\ref{thm:optimum_preservation}--\ref{thm:basin_preservation} is the operative guarantee for a given problem. Single bit-flip neighbours of $\zstar$ are non-optimal configurations, so the minimum over them upper-bounds the minimum over all non-optimal configurations:

\begin{equation}\label{eq:gap-vs-barrier}
    \Delta E^\star \leq \delta_f(\zstar)
\end{equation}

\noindent with equality iff the second-best configuration is a Hamming-distance-$1$ neighbour of $\zstar$. The basin condition is therefore strictly weaker (and so easier to satisfy) than the gap condition, but yields the strictly weaker conclusion of local rather than global optimality. In structured CFNs with multiple near-optimal solutions separated by many bit flips, $\delta_f(\zstar) \gg \Delta E^\star$ is typical, and the basin condition can be satisfied for $\kmax$ values that fall well short of preserving the global optimum but for which $\zstar$ remains a local attractor of single bit-flip descent. In this regime, TBE composes effectively with downstream local-search refinement (Section \ref{sec:algorithm}).

\section{Spectral Perspective and the Noise Floor Condition for Basin Preservation}
\label{sec:spectral}

The bounds of Section \ref{sec:preservation} treat each Ising monomial as an independent contribution to the truncation error and saturate at the worst-case all-ones configuration. This neglects the cancellations between modes of opposite sign that determine the typical, as opposed to worst-case, behaviour of the truncated landscape. In this section we develop a complementary spectral perspective, which makes these cancellations explicit through the Walsh--Hadamard decomposition, and offers a more pragmatic view for determining appropriate $\kmax$.

\subsection{TBE as Orthogonal Projection}
\label{subsec:tbe-spectral}

By Definition~\ref{def:tbe} and the Walsh identification of Equation \eqref{eq:walsh_identification}, the TBE Walsh coefficients are:

\begin{equation}\label{eq:tbe_walsh_coefficients}
    \hat{f}_{\leq \kmax}(S) = \begin{cases} c_S & \text{if } |S| \leq \kmax \\ 0 & \text{if } |S| > \kmax \end{cases}
\end{equation}

\noindent so that TBE is the orthogonal projection of $f$ onto $V_{\leq \kmax}$. The truncation error lies in the orthogonal complement:

\begin{equation}\label{eq:trunc-error}
    \eta(\vec{z}) \coloneqq f(\vec{z}) - f_{\leq \kmax}(\vec{z}) = \sum_{|S| > \kmax} c_S\,\chi_S(\vec{z})
\end{equation}

\noindent and has $L^2$ norm under the uniform measure on the hypercube:

\begin{equation}\label{eq:trunc-var}
    \|\eta\|_2^2 = \E_{\vec{z}}\bigl[\eta(\vec{z})^2\bigr] = \sum_{k > \kmax} P_k \eqqcolon P_{>\kmax}
\end{equation}

\noindent The $L^2$ truncation amplitude $\sqrt{P_{>\kmax}}$ is the natural spectral counterpart of the $L^\infty$ bound $\varepsilon(\kmax)$. The two are related by the elementary inequality:

\begin{equation}\label{eq:l2-vs-linfty}
    \sqrt{P_{>\kmax}} \leq \varepsilon(\kmax) \leq \sqrt{\Nomit\,P_{>\kmax}}
\end{equation}

\noindent where $\Nomit \coloneqq 2^n - \sum_{k=0}^{\kmax} \binom{n}{k}$ counts the omitted Walsh modes. The $L^\infty$ bound is therefore typically a substantial overestimate of the $L^2$ truncation amplitude.

The chain of Equation \eqref{eq:l2-vs-linfty} is the standard $\ell^2 \leq \ell^1 \leq \sqrt{\Nomit}\,\ell^2$ chain on the coefficient vector $(c_S)_{|S|>\kmax}$. The lower bound expresses $\|x\|_2 \leq \|x\|_1$ (with equality when only one coefficient is nonzero), and the upper bound follows from Cauchy--Schwarz applied to $\sum 1 \cdot |c_S|$.

\subsection{The Noise Floor Condition: Brief Overview}
\label{subsec:noise-floor}

We now compare the typical truncation amplitude $\sqrt{P_{>\kmax}}$ to the typical amplitude of the surviving landscape, which an analogous calculation gives as $\sqrt{P_{\leq \kmax}}$. We use \emph{noise floor} in the spectral sense: the truncation residual $\eta$ acts as a low-amplitude perturbation that the surviving landscape $f_{\leq \kmax}$ ``rises above,'' analogously to a signal detectable above noise.

\begin{definition}[Truncation noise floor condition]\label{def:noise_floor}

The cutoff $\kmax$ satisfies the \textit{truncation noise floor condition} if:

\begin{equation}\label{eq:noise_floor}
    P_{\leq \kmax} \gg P_{>\kmax} \quad \text{where} \quad P_{\leq \kmax} = \sum_{k=1}^{\kmax} P_k
\end{equation}

\end{definition}

When Equation \eqref{eq:noise_floor} holds, the truncation removes a spectrally weak perturbation while preserving a spectrally dominant landscape. We now make the qualitative claim of basin preservation precise via a Gaussian-field argument.

\subsection{The Noise Floor Condition: A Gaussian-Field Argument}
\label{subsec:gaussian}

At a generic spin configuration, the signs $\chi_S(\vec{z})$ are roughly balanced: under the uniform measure on $\{-1,+1\}^n$, each $\chi_S$ is a Rademacher random variable with $\E[\chi_S \chi_T] = \delta_{S,T}$ by Equation \eqref{eq:walsh-ortho}, and partial cancellations between modes of opposite sign reduce the magnitude of $\eta(\vec{z})$ from $O(\Nomit\,|c|_{\mathrm{typ}}) = O(\varepsilon(\kmax))$ to $O(\sqrt{\Nomit}\,|c|_{\mathrm{typ}}) = O(\sqrt{P_{>\kmax}})$. We therefore seek to turn the heuristic of Definition~\ref{def:noise_floor} into a checkable condition on the spectral profile. To do so, we move to a \textit{probabilistic ensemble} on the couplings, in the spirit of disordered-system models in statistical physics \cite{Talagrand2011}.

\begin{definition}[Random-coupling ensemble]\label{def:ensemble}

    Fix $\kmax$. The \textit{random-coupling ensemble} draws the omitted couplings $\{c_S : |S| > \kmax\}$ as independent zero-mean random variables with variances $\E[c_S^2] = \pi_S$ and with uniformly bounded fourth moments $\E[c_S^4] \leq C \pi_S^2$ for some constant $C$. The kept couplings $\{c_S : |S| \leq \kmax\}$ may be either deterministic or drawn from a parallel distribution under analogous assumptions. The \textit{expected spectral powers} are:

    \begin{equation}\label{eq:expected-powers}
        \E[P_k] = \sum_{|S| = k} \pi_S, \qquad \E[P_{>\kmax}] = \sum_{k > \kmax} \E[P_k]
    \end{equation}

\end{definition}

This ensemble captures the generic behaviour of truncation residuals on landscapes derived from randomized CFNs; the variance profile $\{\pi_S\}$ is a free input and may be specialized to match any expected spectral shape. The bounded fourth moment is mild: it is satisfied by any sub-Gaussian or bounded distribution on the $c_S$, and is what we use below to convert a max-variance condition into a useful central limit theorem.

\begin{theorem}[Pointwise Gaussian limit for the truncation residual]\label{prop:gaussian_pointwise}

    Let $\vec{z} \in \{-1,+1\}^n$ be fixed, let the omitted couplings be drawn from the ensemble of Definition~\ref{def:ensemble}, and assume the max-variance condition:

    \begin{equation}\label{eq:max-var}
        \frac{\max_{|S| > \kmax} \pi_S}{\E[P_{>\kmax}]} \to 0
    \end{equation}

    \noindent as the number of omitted modes grows. Then the truncation residual $\eta(\vec{z}) = \sum_{|S| > \kmax} c_S \chi_S(\vec{z})$ has variance $\E[P_{>\kmax}]$ independent of $\vec{z}$, and converges to a centered Gaussian:

    \begin{equation}\label{eq:gaussian-amplitude}
        \eta(\vec{z}) \xrightarrow{d} \mathcal{N}\bigl(0,\,\E[P_{>\kmax}]\bigr)
    \end{equation}

\end{theorem}

\begin{proof}

    Fix $\vec{z} \in \{-1,+1\}^n$ and write $X_S \coloneqq c_S\,\chi_S(\vec{z})$ and $s_{\Nomit}^2 \coloneqq \E[P_{>\kmax}]$, so that $\eta(\vec{z}) = \sum_{|S| > \kmax} X_S$. Since $\chi_S(\vec{z}) \in \{-1,+1\}$ is deterministic, each $X_S$ is a Borel function of the single random variable $c_S$, so $\{X_S\}$ is an independent mean-zero family with $\Var[X_S] = \chi_S(\vec{z})^2\,\pi_S = \pi_S$. Summing gives:

    \begin{equation}\label{eq:eta-var-explicit}
        \Var[\eta(\vec{z})] = \sum_{|S| > \kmax} \pi_S = s_{\Nomit}^2
    \end{equation}

    \noindent independently of $\vec{z}$. We now verify Lyapunov's condition at order $\delta = 2$ \cite[Thm.~27.3]{Billingsley1995}. Using $\chi_S(\vec{z})^4 = 1$ and the bounded-fourth-moment hypothesis $\E[c_S^4] \leq C\pi_S^2$:

    \begin{equation}\label{eq:lyapunov-bound}
        \frac{1}{s_{\Nomit}^4}\sum_{|S| > \kmax} \E[X_S^4] \leq \frac{C}{s_{\Nomit}^4}\sum_{|S| > \kmax} \pi_S^2 \leq \frac{C\,\pi_{\max}}{s_{\Nomit}^2}
    \end{equation}

    \noindent where $\pi_{\max} = \max_{|S| > \kmax} \pi_S$ and the second inequality factors out the largest variance: $\sum \pi_S^2 \leq \pi_{\max} \sum \pi_S = \pi_{\max}\,s_{\Nomit}^2$. The right-hand side tends to zero by Equation \eqref{eq:max-var}, so Lyapunov's condition is satisfied.

    Applying Markov's inequality termwise on $\{|X_S| > \varepsilon\,s_{\Nomit}\}$ converts the Lyapunov condition to the Lindeberg condition:

    \begin{align}\label{eq:lindeberg-bound}
        \frac{1}{s_{\Nomit}^2}\sum_{|S| > \kmax} &\E\bigl[X_S^2\,\mathbf{1}\{|X_S| > \varepsilon\,s_{\Nomit}\}\bigr] \nonumber \\ &\leq \frac{1}{\varepsilon^2 s_{\Nomit}^4}\sum_{|S| > \kmax} \E[X_S^4] \to 0
    \end{align}

    \noindent for every $\varepsilon > 0$. The Lindeberg--Feller central limit theorem (see Ref.~\cite[Thm.~3.4.5]{Durrett2019}) then gives $\eta(\vec{z})/s_{\Nomit} \xrightarrow{d} \mathcal{N}(0,1)$, equivalent to Equation \eqref{eq:gaussian-amplitude}.

\end{proof}

Theorem \ref{prop:gaussian_pointwise} establishes the noise scale of truncation: at any fixed configuration, the residual is centered Gaussian with standard deviation $\sqrt{\E[P_{>\kmax}]}$ rather than the worst-case $\varepsilon(\kmax)$. To translate this into a basin-preservation statement, we recall that basin preservation at $\zstar$ requires that each bit-flip energy difference $f(\bitflip{i}) - f(\zstar) \geq 0$ retain its sign after truncation, which holds whenever the residual bit-flip difference $\eta(\bitflip{i}) - \eta(\zstar)$ is smaller in magnitude than the corresponding bit-flip difference of $f_{\leq \kmax}$. Lemma \ref{lem:barrier_variance} provides the variance of the latter, fixing its typical magnitude.

\begin{lemma}[Single bit-flip variance]\label{lem:barrier_variance}

    Under Definition~\ref{def:ensemble} applied also to $\{c_S : |S| \leq \kmax\}$, the
    single bit-flip energy difference of $f_{\leq \kmax}$ at coordinate $i$ has variance:

    \begin{equation}\label{eq:bit-flip-variance}
        v_i \coloneqq \Var\bigl[f_{\leq \kmax}(\bitflip{i}) - f_{\leq \kmax}(\vec{z})\bigr]
        = 4 \sum_{\substack{S \ni i \\ |S| \leq \kmax}} \pi_S
    \end{equation}

    \noindent independently of $\vec{z}$. The average over coordinates satisfies:

    \begin{equation}\label{eq:bit-flip-average}
        \frac{1}{n}\sum_{i=1}^n v_i
        = \frac{4}{n}\sum_{|S| \leq \kmax} |S|\,\pi_S
        \leq \frac{4\,\kmax\,\E[P_{\leq \kmax}]}{n}
    \end{equation}

    \noindent and we denote the per-coordinate floor by $v_{\min} \coloneqq \min_{i \in [n]} v_i$.

\end{lemma}

\begin{proof}

    For any $S \subseteq [n]$, the Walsh function $\chi_S(\vec{z}) = \prod_{j \in S} z_j$ depends on coordinate $i$ if and only if $i \in S$, and negating $z_i$ negates the product:

    \begin{equation}\label{eq:chi-flip}
        \chi_S(\bitflip{i}) = \begin{cases} -\chi_S(\vec{z}) & \text{if } i \in S \\ \phantom{-}\chi_S(\vec{z}) & \text{if } i \notin S \end{cases}
    \end{equation}

    \noindent Each term in the Walsh expansion of $f_{\leq \kmax}$ therefore contributes $c_S[\chi_S(\bitflip{i}) - \chi_S(\vec{z})] = -2 c_S \chi_S(\vec{z})$ when $i \in S$ and zero otherwise:

    \begin{equation}\label{eq:bit-flip-diff}
        f_{\leq \kmax}(\bitflip{i}) - f_{\leq \kmax}(\vec{z}) = -2 \sum_{\substack{S \ni i \\ |S| \leq \kmax}} c_S \chi_S(\vec{z})
    \end{equation}

    \noindent The summands are deterministic multiples of distinct independent couplings, hence themselves independent, with variance $4\,\chi_S(\vec{z})^2\,\pi_S = 4\,\pi_S$ each. Summing the variances thus gives Equation \eqref{eq:bit-flip-variance}, independent of $\vec{z}$ since the sign factors $\chi_S(\vec{z})$ are eliminated by squaring.

    Then, summing Equation \eqref{eq:bit-flip-variance} over $i$ and exchanging the order of summation:

    \begin{align}\label{eq:total-sum-exchange}
        \sum_{i=1}^n v_i = \sum_{i=1}^n \sum_{\substack{S \ni i \\ |S| \leq \kmax}} 4\,\pi_S &= 4 \sum_{|S| \leq \kmax} \pi_S\,\bigl|\{i \in [n] : i \in S\}\bigr| \nonumber \\ &= 4 \sum_{|S| \leq \kmax} |S|\,\pi_S
    \end{align}

    \noindent since each $S$ contributes once for each of its $|S|$ elements. Bounding $|S| \leq \kmax$ pointwise in the sum gives $\sum_{|S| \leq \kmax} |S|\,\pi_S \leq \kmax\,\E[P_{\leq \kmax}]$, which on dividing by $n/4$ yields Equation \eqref{eq:bit-flip-average}.

\end{proof}

Equation \eqref{eq:bit-flip-average} reveals how system size enters. The average per-coordinate bit-flip variance is at most $4\kmax \E[P_{\leq \kmax}]/n$, so the typical single bit-flip difference at a representative coordinate scales as $\sqrt{\kmax\,\E[P_{\leq\kmax}]/n}$. Combining this with the residual scale from Theorem \ref{prop:gaussian_pointwise} yields the noise floor condition detailed in Theorem \ref{thm:strong_noise_floor}.

\begin{theorem}[Strong noise floor condition]\label{thm:strong_noise_floor}

    Under the ensemble of Definition~\ref{def:ensemble} applied to all couplings, suppose:

    \begin{equation}\label{eq:strong-noise-floor}
        \frac{\E[P_{>\kmax}]}{\E[P_{\leq \kmax}]} \ll \frac{\kmax}{n}
    \end{equation}

    \noindent Then with high probability over the coupling draw, at any fixed configuration $\vec{z}$ and any typical coordinate $i$:

    \begin{equation}\label{eq:sign-preservation}
        \mathrm{sign}\bigl[f(\bitflip{i}) - f(\vec{z})\bigr] = \mathrm{sign}\bigl[f_{\leq \kmax}(\bitflip{i}) - f_{\leq \kmax}(\vec{z})\bigr]
    \end{equation}

\end{theorem}

\begin{proof}

    The full function bit-flip difference decomposes as:

    \begin{equation}\label{eq:diff-decomp}
        f(\bitflip{i}) - f(\vec{z}) = \bigl[f_{\leq \kmax}(\bitflip{i}) - f_{\leq \kmax}(\vec{z})\bigr] + \bigl[\eta(\bitflip{i}) - \eta(\vec{z})\bigr]
    \end{equation}

    \noindent The two quantities in Equation \eqref{eq:sign-preservation} therefore differ exactly by the noise term, and sign preservation holds whenever the noise magnitude is smaller than the signal magnitude.

    \textit{Noise scale.} The same calculation as Lemma \ref{lem:barrier_variance}, applied to $\eta$ in place of $f_{\leq \kmax}$ (with the sum restricted to $|S| > \kmax$ rather than $|S| \leq \kmax$), gives:

    \begin{equation}\label{eq:noise-var}
        \Var\bigl[\eta(\bitflip{i}) - \eta(\vec{z})\bigr] = 4\sum_{\substack{S \ni i \\ |S| > \kmax}} \pi_S \leq 4\,\E[P_{>\kmax}]
    \end{equation}

    \noindent uniformly in $i$. The residual bit-flip difference is itself a sum over independent couplings with deterministic Walsh-sign coefficients, so Theorem \ref{prop:gaussian_pointwise} applies in the same form. Gaussian tail bounds then give:

    \begin{equation}\label{eq:noise-tail}
        \bigl|\eta(\bitflip{i}) - \eta(\vec{z})\bigr| = O\bigl(\sqrt{\E[P_{>\kmax}]}\bigr)
    \end{equation}

    \noindent with high probability over the coupling draw.

    \textit{Signal scale.} By Lemma \ref{lem:barrier_variance} applied to $f_{\leq \kmax}$, the total bit-flip variance summed over coordinates is at most $4\kmax\,\E[P_{\leq \kmax}]$, so the average per-coordinate variance is at most $4\kmax\,\E[P_{\leq \kmax}]/n$. At a typical $i$:

    \begin{equation}\label{eq:signal-scale}
        \bigl|f_{\leq \kmax}(\bitflip{i}) - f_{\leq \kmax}(\vec{z})\bigr| = \Omega\bigl(\sqrt{\kmax\,\E[P_{\leq \kmax}]/n}\bigr)
    \end{equation}

    \noindent with high probability. Sign preservation in Equation \eqref{eq:diff-decomp} holds whenever the noise scale of Equation \eqref{eq:noise-tail} is dominated by the signal scale of Equation \eqref{eq:signal-scale}:

    \begin{equation}\label{eq:scale-comparison}
        \sqrt{\E[P_{>\kmax}]} \ll \sqrt{\kmax\,\E[P_{\leq \kmax}]/n}
    \end{equation}

    \noindent Squaring both sides and dividing by $\E[P_{\leq \kmax}]$ yields Equation \eqref{eq:strong-noise-floor}.

\end{proof}

\begin{corollary}[Heuristic basin preservation]\label{cor:basin}

    Under Theorem \ref{thm:strong_noise_floor}, with high probability over the coupling draw, $\zstar$ remains a single bit-flip local minimum of $f_{\leq \kmax}$ whenever it is a single bit-flip local minimum of $f$, and basin assignments under single bit-flip descent on $f$ and $f_{\leq \kmax}$ coincide on a high-probability fraction of the hypercube.

\end{corollary}

\begin{proof}

    Apply Theorem \ref{thm:strong_noise_floor} at $\vec{z} = \zstar$ and each coordinate $i \in [n]$; sign preservation of the bit-flip difference is precisely the statement that $f_{\leq \kmax}(\vec{z}^{\star\oplus i}) - f_{\leq \kmax}(\zstar)$ has the same sign as $f(\vec{z}^{\star\oplus i}) - f(\zstar) \geq 0$, so $\zstar$ remains a local minimum. Iterating over neighbouring configurations gives the claim.

\end{proof}

The argument is heuristic in three respects, each setting a direction for strengthening: the Gaussian approximation~\eqref{eq:gaussian-amplitude} holds pointwise rather than uniformly; Theorem \ref{thm:strong_noise_floor} conditions on a typical rather than worst-case coordinate $i$ and so does not yield a uniform conclusion over $[n]$; and the conclusion is probabilistic over the coupling ensemble rather than deterministic for any given problem. We claim it only as a heuristic and leave further rigour to future work. The remaining question for us now is: what features of the underlying CFN produce a favourable spectral profile in the first place?

\section{Cost Function Network Structure and Spectral Decay Favourable to TBE}
\label{sec:cfn-structure}

The noise floor condition controls when a given exact-binary HUBO admits a small-$\kmax$ truncation, but does not specify which CFN structures produce a favourable spectral profile. We address that question by expressing the Walsh coefficients $c_S$ of the encoded HUBO directly in terms of the Walsh transforms of the underlying cost tables $\alpha_i$ and $\beta_{i,j}$ taken on their respective sub-hypercubes. The spectral profile decomposes additively over cost tables, and the question of when $P_k$ decays rapidly reduces to the well-studied question of when the Walsh transform of a function on the hypercube is concentrated at low degree \cite{ODonnell2014, KahnKalaiLinial1988}. The connection between low-degree concentration and smoothness is mediated by a discrete calculus on the hypercube, which we develop in Section \ref{subsec:discrete-calc} after the cost-table identification.

\subsection{Walsh Coefficients of the Encoded HUBO}
\label{subsec:cfn-walsh}

The indicator polynomials of Equation \eqref{eq:indicator_walsh} are already exhibited as Walsh expansions on the sub-hypercubes induced by their registers, so the encoded HUBO of Equation \eqref{eq:ising_hubo} is also already an aggregate of per-register Walsh expansions. Before reading off its couplings, we adopt a normalization convention that keeps unary and pairwise contributions in disjoint Walsh subspaces.

Each pairwise cost table $\beta_{i,j}$ is taken to satisfy:

\begin{equation}\label{eq:centered_norm}
    \sum_{c_j} \beta_{i,j}(c_i, c_j) = 0 \quad \forall c_i,
    \quad
    \sum_{c_i} \beta_{i,j}(c_i, c_j) = 0 \quad \forall c_j
\end{equation}

\noindent any nonzero marginals having been absorbed into $\alpha_i$ and $\alpha_j$. This entails no loss of generality, and under it the Walsh modes of $\widetilde{\beta}_{i,j}$ supported on a single register vanish. (Note that $\sigma_\emptyset = 1$, so the marginals in Equation \eqref{eq:centered_norm} are, up to a $2^{-D}$ prefactor, precisely $\hat\beta_{i,j}(\emptyset, T_j)$ and $\hat\beta_{i,j}(T_i, \emptyset)$.)

Extend each unary cost table by zero on unused bitstrings:

\begin{equation}\label{eq:alpha_tilde}
    \widetilde{\alpha}_i(\vec{z}_i) = \begin{cases}
        \alpha_i(c_i) & \text{if } \vec{z}_i = \vec{s}^{\,(c_i)} \\
        0 & \text{otherwise}
    \end{cases}
\end{equation}

\noindent and define $\widetilde{\beta}_{i,j}$ analogously. With these conventions in hand:

\begin{theorem}[Cost-table form of the encoded couplings]\label{thm:walsh_coeffs}

    Under the centered normalization of Equation \eqref{eq:centered_norm}, let $S \subseteq [n]$ be nonempty. Then:
    
    \begin{enumerate}[leftmargin=*]

        \item If $S \subseteq [D_i]$ for a unique register $i$,
        \begin{equation}\label{eq:J_unary}
            c_S = \hat{\alpha}_i(S)
                = \frac{1}{2^{D_i}} \sum_{c_i = 1}^{|d_i|}
                  \alpha_i(c_i)\,\sigma_S^{(c_i)}.
        \end{equation}

        \item If $S = T_i \sqcup T_j$ with $T_i \subseteq [D_i]$ and $T_j \subseteq [D_j]$
        both nonempty,
        \begin{equation}\label{eq:J_binary}
            c_S = \hat{\beta}_{i,j}(T_i, T_j)
                = \frac{1}{2^{D_i + D_j}}
                  \sum_{c_i, c_j} \beta_{i,j}(c_i, c_j)\,
                  \sigma_{T_i}^{(c_i)}\,\sigma_{T_j}^{(c_j)}.
        \end{equation}

        \item If $S$ intersects three or more registers, $c_S = 0$.
    \end{enumerate}

\end{theorem}

\begin{proof}

    First, observe that $\sigma_T^{(c_i)} = \chi_T(\vec{s}^{\,(c_i)})$, which makes Equation \eqref{eq:indicator_walsh} the Walsh expansion of the indicator restricted to the encoding image. Composing with the cost-table values via Equations~\eqref{eq:unary-encoded}--\eqref{eq:binary-encoded} and reading off the coefficient of $\chi_T$:

    \begin{equation}\label{eq:alpha-walsh-explicit}
        \hat{\widetilde{\alpha}}_i(T)=\frac{1}{2^{D_i}} \sum_{c_i} \alpha_i(c_i)\,\sigma_T^{(c_i)}
    \end{equation}

    \noindent and:

    \begin{equation}\label{eq:beta-walsh-explicit}
        \hat{\widetilde{\beta}}_{i,j}(T_i, T_j) = \frac{1}{2^{D_i + D_j}} \sum_{c_i, c_j} \beta_{i,j}(c_i, c_j)\,\sigma_{T_i}^{(c_i)}\,\sigma_{T_j}^{(c_j)}
    \end{equation}

    \noindent The pairwise expansion supplies Walsh modes whose support $S = T_i \cup T_j$ may \textit{a priori} have $T_i = \emptyset$ or $T_j = \emptyset$, in which case the support lies entirely on a single register. 
    
    Under the centered normalization of Equation \eqref{eq:centered_norm}, however, the inner sums in Equation \eqref{eq:beta-walsh-explicit} reduce to a marginal sum that vanishes whenever $T_i = \emptyset$ or $T_j = \emptyset$, since $\sigma_\emptyset = 1$. Consequently, the only Walsh modes of $\widetilde{\beta}_{i,j}$ that contribute have $T_i, T_j$ both nonempty; this gives case~(2), and case~(1) reduces to the pure unary contribution $\hat{\widetilde{\alpha}}_i$.

    Case~(3) is forced by the pairwise structure of $f_{\mathrm{CFN}}$: no cost table is supported on more than two registers, so no Walsh mode of $f$ straddles three or more.

\end{proof}

The immediate consequence of interest to us is the additive form of the CFN spectral profile.

\begin{corollary}[Additive spectral decomposition]\label{cor:spectral_decomp}

    Let:

    \begin{equation}\label{eq:Pk-unary}
        P_k^{(i)} = \sum_{\substack{T \subseteq [D_i] \\ |T| = k}} \hat{\widetilde{\alpha}}_i(T)^2
    \end{equation}

    \noindent denote the degree-$k$ unary spectral power and:

    \begin{equation}\label{eq:Pk-binary}
        P_k^{(i,j)} = \sum_{\substack{|T_i| + |T_j| = k \\ T_i, T_j \neq \emptyset}} \hat{\widetilde{\beta}}_{i,j}(T_i, T_j)^2
    \end{equation}

    \noindent the bidegree-$k$ pairwise spectral power. Then:

    \begin{equation}\label{eq:additive_Pk}
        P_k = \sum_{i=1}^{N} P_k^{(i)} + \sum_{1\leq i < j \leq N} P_k^{(i,j)}
    \end{equation}

    \noindent The unary contribution $P_k^{(i)}$ vanishes for $k > D_i$, and the binary contribution $P_k^{(i,j)}$ vanishes for $k > D_i + D_j$ or $k \leq 1$.

\end{corollary}

Equation \eqref{eq:additive_Pk} is the reason CFN-level structural arguments suffice for determining when the noise-floor condition should be satisfied: each cost table contributes orthogonally to the global spectrum, and the noise floor condition holds for $f$ if and only if it holds in aggregate across the cost-table spectra. This reduces the global question ``which CFNs are TBE-amenable?'' to the local question ``which cost tables, under which bitstring assignments, have front-loaded Walsh spectra?'' This question we begin to address by way of discrete calculus of the next subsection.

\subsection{Discrete Calculus on the Hypercube}
\label{subsec:discrete-calc}

The notion of \textit{smoothness} on the hypercube $\{-1,+1\}^D$ is captured by a discrete analogue of partial differentiation (see Ref.~\cite[Ch.~2]{ODonnell2014}), and bounds on its higher derivatives translate directly into bounds on the high-degree Walsh spectrum. Because we use this machinery heavily to convert qualitative claims about ``smooth potentials'' into quantitative spectral statements, we develop it carefully here.

For a function $f\colon \{-1,+1\}^D \to \mathbb{R}$ and a coordinate $q \in [D]$, write $\vec{z}^{q\to+}$ and $\vec{z}^{q\to-}$ for the configurations obtained from $\vec{z}$ by setting coordinate $q$ to $+1$ or $-1$ respectively (leaving all other coordinates unchanged). The \textit{discrete partial derivative of $f$ along coordinate $q$} is then:

\begin{equation}\label{eq:discrete-partial}
    D_{\{q\}} f(\vec{z}) \coloneqq \tfrac{1}{2}\bigl(f(\vec{z}^{q\to+}) - f(\vec{z}^{q\to-})\bigr)
\end{equation}

\noindent That is, $D_{\{q\}} f(\vec{z})$ is the average rate of change of $f$ as coordinate $q$ flips from $-1$ to $+1$, divided by the coordinate spacing $(+1) - (-1) = 2$.

Then, for any $T \subseteq [D]$, the Walsh function $\chi_T(\vec{z}) = \prod_{i \in T} z_i$ depends on coordinate $q$ if and only if $q \in T$. If $q \notin T$, then $\chi_T(\vec{z}^{q\to+}) = \chi_T(\vec{z}^{q\to-})$, so $D_{\{q\}} \chi_T = 0$. If $q \in T$, then $\chi_T(\vec{z}^{q\to\pm}) = (\pm 1)\,\chi_{T \setminus \{q\}}(\vec{z})$, so the partial derivative is then:

\begin{equation}\label{eq:walsh-derivative}
    D_{\{q\}} \chi_T = \begin{cases} \chi_{T \setminus \{q\}} & \text{if } q \in T \\ 0 & \text{if } q \notin T \end{cases}
\end{equation}

\noindent Thus single-coordinate partial differentiation acts on the Walsh basis by removing coordinate $q$ from the index set when present and annihilating the mode otherwise.

For $S \subseteq [D]$, the \textit{$|S|$th-order mixed partial derivative} of $f$ along coordinates $S$ is defined by composition:

\begin{equation}\label{eq:DS-as-product}
    D_S f \coloneqq \prod_{q \in S} D_{\{q\}} f
\end{equation}

\noindent The single-coordinate operators $D_{\{q\}}$ commute (each independently fixes its own coordinate), so the order of composition does not matter and Equation \eqref{eq:DS-as-product} is well-defined. Iterating Equation \eqref{eq:walsh-derivative}, we obtain the action on Walsh modes:

\begin{equation}\label{eq:DS-walsh-action}
    D_S \chi_T = \begin{cases} \chi_{T \setminus S} & \text{if } S \subseteq T \\ 0 & \text{otherwise} \end{cases}
\end{equation}

\noindent which combined with the Walsh expansion $f = \sum_T \hat{f}(T) \chi_T$ gives the following:

\begin{equation}\label{eq:DSg}
    D_S f(\vec{z}) = \sum_{T \supseteq S} \hat{f}(T)\,\chi_{T \setminus S}(\vec{z})
\end{equation}

\noindent Concretely, $D_S f$ is the function that retains only those Walsh modes whose support contains every coordinate in $S$, and shifts each surviving mode down by $|S|$ in degree.

\begin{definition}[Lipschitz constants and Parseval's identity]\label{def:lipschitz}

    The $k$th-order \textit{discrete Lipschitz constant} of $f$ is:

    \begin{equation}\label{eq:Lk-def}
        L_k(f) \coloneqq \max_{|S| = k} \|D_S f\|_\infty
    \end{equation}

    \noindent which measures the largest possible magnitude that any $k$th-order discrete derivative of $f$ can attain anywhere on the hypercube.

\end{definition}

By Parseval's identity applied to Equation \eqref{eq:DSg}, the squared $L^2$ norm of the $k$th derivative is the sum of squared Walsh coefficients of $f$ at all index sets containing $S$:

\begin{equation}\label{eq:DSg_l2}
    \|D_S f\|_2^2 = \E_{\vec{z}}\bigl[\left(D_S f(\vec{z})\right)^2\bigr] = \sum_{T \supseteq S} \hat{f}(T)^2
\end{equation}

\noindent This is the bridge between smoothness (bounds on $L_k$) and spectral concentration (bounds on $P_{>k}$). We develop this further in the following subsection, where we see that smoothness implies favourable spectral decay.

\subsection{Smoothness Implies Spectral Decay}
\label{subsec:smoothness_decay}

We can now state the precise sense in which bounded high-order derivatives imply rapidly decaying Walsh tails.

\begin{lemma}[Tail bound from $k$th-order smoothness]\label{lem:tail_bound}

For any $f\colon \{-1,+1\}^D \to \mathbb{R}$ and any $k \geq 1$:

\begin{equation}\label{eq:tail_bound}
    \sum_{j \geq k} \binom{j}{k} P_j(f) = \sum_{|S| = k} \|D_S f\|_2^2 \leq \binom{D}{k}\,L_k(f)^2
\end{equation}

\noindent and consequently:

\begin{equation}\label{eq:Pk_and_tail}
    P_k(f) \leq \binom{D}{k}\,L_k(f)^2 \quad \text{and} \quad \sum_{j \geq k} P_j(f) \leq \binom{D}{k}\,L_k(f)^2
\end{equation}

\end{lemma}

\begin{proof}

Sum Equation \eqref{eq:DSg_l2} over $S$ with $|S| = k$. Each Walsh coefficient $\hat{f}(T)^2$ with $|T| = j \geq k$ appears once for each $S \subseteq T$ with $|S| = k$, i.e.\ $\binom{j}{k}$ times:

\begin{equation}\label{eq:sum-over-S}
    \sum_{|S| = k} \|D_S f\|_2^2 = \sum_{|S| = k} \sum_{T \supseteq S} \hat{f}(T)^2 = \sum_{j \geq k} \binom{j}{k} P_j(f)
\end{equation}

\noindent Bounding $\|D_S f\|_2^2 \leq \|D_S f\|_\infty^2 \leq L_k(f)^2$ for each $S$ and using $|\{S : |S| = k\}| = \binom{D}{k}$ gives the right-hand inequality of Equation \eqref{eq:tail_bound}. The two consequences in Equation \eqref{eq:Pk_and_tail} follow because $\binom{j}{k} \geq 1$ for $j \geq k$, so each $P_j(f)$ on the left side of Equation \eqref{eq:tail_bound} is bounded individually and the tail sum is bounded by dropping the binomial weights.

\end{proof}

The interpretation of Lemma \ref{lem:tail_bound} is that controlling $L_k(f)$ controls the Walsh mass at all degrees $j \geq k$. If $L_k(f) = O(c^k)$ for some $c < 1$, then since $\binom{D}{k} \leq 2^D$, the tail $\sum_{j \geq k} P_j(f)$ decays at least geometrically in $k$, and the noise floor condition holds for any $\kmax$ above the constant where $L_k$ becomes small relative to lower-order derivatives. If instead $L_k(f) = O(c^k)$ for $c \geq 1$, the noise floor condition cannot be verified through smoothness alone.

\subsection{Practical Diagnostic Workflow for TBE Efficacy}
\label{subsec:cfn_summary}

The combination of Corollary \ref{cor:spectral_decomp} and Lemma \ref{lem:tail_bound} converts the abstract noise floor condition into a checkable structural criterion. A pairwise CFN $f_{\mathrm{CFN}}$ is a candidate for small-$\kmax$ TBE when, under the chosen bitstring assignment:

\begin{enumerate}[leftmargin=*]

    \item Each unary cost table $\widetilde{\alpha}_i$ has small high-order Lipschitz constants $L_k(\widetilde{\alpha}_i)$, decaying geometrically (or vanishing identically beyond a fixed degree);

    \item Each pairwise cost table $\widetilde{\beta}_{i,j}$ likewise has small $L_k(\widetilde{\beta}_{i,j})$, with the same decay structure; and

    \item The interaction graph is sparse, or, equivalently, the per-edge contributions $P_k^{(i,j)}$ aggregate to a controlled total relative to the unary contributions $P_k^{(i)}$.

\end{enumerate}

Each $P_k^{(i)}$ and $P_k^{(i,j)}$ can be precomputed via a fast Walsh--Hadamard transform on the corresponding sub-hypercube, in time $O(|d_i|\,D_i)$ for unary tables and $O(|d_i|\,|d_j|\,(D_i + D_j))$ for binary tables; aggregating these per-table profiles via Equation \eqref{eq:additive_Pk} yields the global spectrum and feeds directly into the noise-floor diagnostic and the $L^\infty$ bound. Per-table profiles also identify spectral outliers: when the global spectrum is dominated by a small set of tables with high-degree mass, those tables are candidates either for replacement (if a smoother surrogate is available), for re-assignment (if their smoothness is being undermined by a poor bitstring choice), or for exclusion from TBE. When all per-table profiles exhibit geometric decay, the additive decomposition guarantees that the assembled HUBO does as well, and the choice of $\kmax$ can be made by inspection of the aggregate profile.

\section{TBE Compilation Algorithm}
\label{sec:algorithm}

Algorithm~\ref{alg:tbe} describes the end-to-end compilation pipeline. Because we work in the Ising basis from the start, no basis conversion is required between obtaining the encoding and computing the spectral profile, and Walsh diagnostics are immediately available from the couplings.

\begin{algorithm}[t]
\caption{Truncated-Binary Encoding Compilation}
\label{alg:tbe}
    \begin{algorithmic}[1]
        \REQUIRE CFN $f_{\mathrm{CFN}}$; cutoff $\kmax$
        \ENSURE Truncated Ising HUBO $f_{\leq \kmax}$; certified error $\varepsilon$
        \STATE Apply exact binary encoding in the Ising basis (Equations~\eqref{eq:indicator}--\eqref{eq:ising_hubo}) to obtain Ising couplings $\{c_S\}$
        \STATE Compute spectral profile $P_k = \sum_{|S|=k} c_S^2$ for $k = 0, \ldots, \kfull$
        \STATE Verify the noise floor condition $P_{\leq \kmax} \gg P_{>\kmax}$ (Definition~\ref{def:noise_floor})
        \STATE $\varepsilon \leftarrow \sum_{|S| > \kmax} |c_S|$ \hfill\textit{(Theorem \ref{thm:linfty})}
        \STATE Set $f_{\leq \kmax} \leftarrow \sum_{|S| \leq \kmax} c_S \prod_{i \in S} z_i$
        \STATE Solve $\arg\min f_{\leq \kmax}$ on quantum hardware (with quadratization to QUBO if required by target)
        \STATE Decode the spin configuration to a CFN configuration; evaluate $f_{\mathrm{CFN}}$ as ground truth
        \STATE Optionally refine via single bit-flip descent on full $f$
        \RETURN Best CFN configuration found, certified error $\varepsilon$
    \end{algorithmic}
\end{algorithm}

\textit{Quadratization for $\kmax > 2$.} If the downstream solver requires QUBO form, the truncated HUBO must still be reduced to $2$-local form via standard ancilla constructions \cite{Boros2002, Dattani2019, Anthony2017}. The cost of this reduction is controlled by $\kmax$ rather than by $\kfull$; choosing $\kmax = 2$ avoids quadratization entirely, while $\kmax \in \{3, 4\}$ produces small, controlled gadgetization overhead. Quadratization for an Ising target proceeds analogously to the standard $\{0,1\}$-form construction, with auxiliary spins playing the role of auxiliary bits.

\textit{Decoding and verification.} Because TBE is approximate, the spin configuration $\ztrunc = \arg\min f_{\leq \kmax}$ recovered from hardware must be re-evaluated against the full CFN cost $f_{\mathrm{CFN}}$ to obtain a reliable score. The unused-bitstring problem of exact binary encoding---arising when $|d_i| < 2^{D_i}$---persists under TBE and is handled identically, by penalty terms or by mapping unused bitstrings to a fallback choice \cite{Berwald2024}.

\textit{Local refinement.} Theorem \ref{thm:basin_preservation} ensures that under the basin-barrier condition, single bit-flip descent on the full $f$ initialized at $\ztrunc$ either fixes immediately (if $\ztrunc$ is a local minimum of $f$) or descends to a nearby local minimum. This makes TBE compose naturally with any downstream local-search refinement, with $\ztrunc$ supplying a warm start whose distance to the nearest local minimum of $f$ is bounded by the basin-barrier slack $\delta_f(\zstar) - 2\varepsilon(\kmax)$.

\section{Discussion and Conclusions}
\label{sec:discussion}

Truncated-binary encoding occupies a previously underexplored point in the space of CFN-to-HUBO compilations. The standard alternatives---one-hot, domain-wall, exact binary---all aim for exact representations and trade only along the qubit-count versus QUBO-graph-density axis \cite{Berwald2024}. TBE operates along a third axis, \textit{approximation quality} (as does the approximate-binary encoding of Ref. \cite{Zaborniak2025}), with the exact encoding recovered as $\kmax \to \kfull$ and the QUBO regime recovered at $\kmax = 2$, and intermediate $\kmax$ recovering lower-degree HUBOs in between.

The Ising-basis formulation we have adopted is strictly necessary. As the spectral leakage formula of Theorem \ref{prop:leakage} makes precise, $\{0,1\}$-form HUBO truncation distributes mass across all lower Walsh degrees; only in the Ising basis does HUBO-degree truncation coincide with Walsh-degree truncation, and only there does the spectral profile $\{P_k\}$ provide a direct readout of the coupling magnitudes. Further, the combination of the noise floor condition (Section \ref{sec:spectral}) with the additive cost-table decomposition (Section \ref{sec:cfn-structure}) traces TBE applicability to specific structural features of the underlying CFN.

Several extensions of the present analysis are immediate. First, the heuristic Gaussian-field argument of Section \ref{subsec:gaussian} can be made more rigorous, replacing pointwise CLT estimates by uniform tail bounds via hypercontractivity \cite[Ch.~9]{ODonnell2014} and treating the residual $\eta$ as a correlated Gaussian field on the hypercube \cite{Talagrand2011}. Second, the choice of bitstring assignment $\vec{r}^{(c_i)}$ is unspecified by exact binary encoding; different assignments for different cost tables may be necessary (\textit{e.g.}, Gray code for smooth physical potentials). Per-table assignment optimization, jointly with $\kmax$, is itself a combinatorial problem and is open. Finally, benchmarking on analog and gate-based quantum hardware implementations across a representative suite of CFNs is required to establish the practical regime of TBE's effectiveness.

\section*{Acknowledgments}

\small TZ was funded by a National Sciences and Engineering Research Council of Canada (NSERC) Collaborative Research and Training Experience (CREATE) grant on Quantum Computing, NSERC Alliance Consortium Grant entitled Quantum Software Consortium -- Exploring Distributed Quantum Solutions for Canada (QSC), an NSERC Alliance grant on Quantum Computing for Optimal Mobility and the Simons Foundation. The author thanks Vikram Khipple Mulligan, Ulrike Stege, and Hausi M\"uller for useful discussions that informed this work.

\bibliographystyle{IEEEtran}
\bibliography{main}

\end{document}